\documentclass[runningheads]{llncs}
\usepackage{amssymb}
\usepackage{amsmath}
\usepackage{comment}
\usepackage[color=gray!27]{todonotes}
\usepackage{inputenc}
\usepackage{shuffle}
\usepackage{multirow}
\usepackage{listings}
\usepackage{mathabx}
\usepackage{hyperref}

\usepackage{shuffle}

\usepackage{tikz}
\usetikzlibrary{positioning,shadows,arrows}
\usetikzlibrary{arrows,automata,positioning}
\usetikzlibrary{shapes,shapes.geometric,arrows,fit,calc,positioning,automata,}

\usepackage{diagbox}
\usepackage{slashbox}
\usepackage{tikz}
\usetikzlibrary{matrix}

\usepackage{apxproof}
\theoremstyle{plain}
\newtheoremrep{theorem}{Theorem}
\newtheoremrep{proposition}[theorem]{Proposition}
\newtheoremrep{lemma}[theorem]{Lemma}
\newtheoremrep{claim}[theorem]{Claim}
\newtheoremrep{conjecture}[theorem]{Conjecture}
\newtheoremrep{corollary}[theorem]{Corollary}
\newtheoremrep{definition}[theorem]{Definition}

\newenvironment{claiminproof}[1]{\medskip\par\noindent\underline{Claim:}\space#1}{}
\newenvironment{claimproof}[1]{\begin{quote}\par\noindent\emph{Proof of the Claim:}\space#1}{[\emph{End, Proof of the Claim}]\end{quote}}

\DeclareFontFamily{U}{bigshuffle}{}
\DeclareFontShape{U}{bigshuffle}{m}{n}{
  <5-8> s*[1.7] shuffle7
  <8->  s*[1.7] shuffle10
}{}
\DeclareSymbolFont{BigShuffle}{U}{bigshuffle}{m}{n}
\DeclareMathSymbol\bigshuffle{\mathop}{BigShuffle}{"001}
\DeclareMathSymbol\bigcshuffle{\mathop}{BigShuffle}{"002}

\newcommand{\Orb}{\operatorname{Orb}}

\newcommand{\PSPACE}{\textsf{PSPACE}}

\begin{document}
%
\title{State Complexity of Projection on Languages Recognized by Permutation Automata and Commuting Letters}
\titlerunning{State Complexity of Projection on Permutation Automata}

%
%
\author{Stefan Hoffmann\orcidID{0000-0002-7866-075X}}
\authorrunning{S. Hoffmann}
%
\institute{Informatikwissenschaften, FB IV, 
  Universit\"at Trier,  Universitätsring 15, 54296~Trier, Germany, 
  \email{hoffmanns@informatik.uni-trier.de}}
\maketitle              
\begin{abstract}

 The projected language of a general deterministic automaton with $n$
 states is recognizable by a deterministic automaton with $2^{n-1} + 2^{n-m} - 1$
 states, where $m$ denotes the number of states incident to unobservable non-loop transitions, and this bound is best possible.
 Here, we derive the tight
 bound $2^{n - \lceil \frac{m}{2} \rceil} - 1$ for permutation automata.
 For a state-partition automaton with $n$ states (also called automata with the observer property) the projected
 language is recognizable with $n$ states.
 Up to now, these, and finite languages projected onto unary languages,
 were the only classes of automata known to possess this property.
 We show that this is also true for commutative automata
 and we find commutative automata
 that are not state-partition automata.
 
\keywords{state complexity \and finite automata  \and projection \and permutation automata \and state-partition automata \and commutative automata} 
\end{abstract}
%
%
%

\section{Introduction}
\label{sec:introduction}

%
%
%
%

The state complexity of a regularity-preserving operation is the minimal number of states 
needed in a recognizing automaton for the result of this operation, dependent on the size
of the input automaton.
The study of the state complexity was initiated in~\cite{Mas70}
and systematically started in~\cite{YuZhuangSalomaa1994}.
As the number of states of a recognizing automaton 
could be interpreted as the memory required to describe the recognized language
and is directly related to the runtime of algorithms employing regular languages, obtaining
state complexity bounds is a natural question with applications in verification, natural language
processing or software engineering~\cite{GaoMRY17,kohavi_jha_2009,mihov_schulz_2019,roche97,DBLP:reference/crc/2012fsbma}.

Here, in terms of state complexity, we are concerned with deterministic automata only. 
There were also investigations using nondeterministic automata~\cite{DBLP:journals/ijfcs/HolzerK09}.
However, deterministic automata have better algorithmic properties: (1) equality could be done in almost linear time~\cite{Hopcroft&Karp:2015},  (2) the minimal automaton is unique up to isomorphism~\cite{HopUll79} and (3) there is an $O(n\log n)$-time minimization
algorithm~\cite{Hopcroft:1971}. Contrary, for nondeterministic automata, equality testing is $\PSPACE$-complete~\cite{stockmeyer1973word}, 
minimal automata are not unique and minimization is a $\PSPACE$-complete problem~\cite{DBLP:journals/ijfcs/HolzerK09}.

The state complexity of the projection operation was investigated in~\cite{DBLP:journals/tcs/JiraskovaM12,Wong98}.
In~\cite{Wong98}, the tight upper bound $3 \cdot 2^{n-2} - 1$
was shown, and in~\cite{DBLP:journals/tcs/JiraskovaM12} the refined, and tight, bound $2^{n-1} + 2^{n-m} - 1$
was shown, where $m$ is related to the number of unobservable transitions for the projection operator.

The projection operator has applications in engineering, verification, fault diagnosis
and supervisory control~\cite{DBLP:books/daglib/0034521,DBLP:journals/automatica/KomendaMS12,DBLP:books/sp/13/KomendaMS13a,wonham2019}, as
it corresponds to the observable behavior, a simplified or a restricted view of a modeled system.
However, as, in general, the resulting automaton could be exponentially large, in practical applications
only those projections that avoid this blow-up are interesting.
Motivated by this, in~\cite{DBLP:conf/ifipTCS/JiraskovaM12} 
state-partition automata for a projection were introduced, a class of automata
for which the projection is recognizable with $n$ states, if the input automaton has $n$ states.

Permutation automata were introduced in~\cite{DBLP:journals/mst/Thierrin68} and by McNaugthon~\cite{McNaughton67} 
in connection with the star-height problem. The languages recognized by permutation
automata are called (pure-)group languages~\cite{McNaughton67,Pin86,DBLP:reference/hfl/Pin97}.
However, one could argue that, if not viewed
as language recognizing devices, but as mere state-transition systems, sometimes also just called
semi-automata, permutation automata were around under the disguise of finite permutation groups, i.e., subgroups
of the group of all permutation on a finite set, since the beginning of the 19th century, starting
with the work of Galois, Lagrange, Jordan and others~\cite{cameron_1999,Neumann2011}.
However, certainly, the viewpoint was different.

Languages recognized by permutation automata
 are not describable by first-order formulae
 using only the order relation~\cite{McNaughton71}
and commutative regular languages correspond to threshold and modulo counting of letters~\cite{DBLP:reference/hfl/Pin97}. 
The languages recognized by certain permutation automata, 
for example whose transformation monoids are solvable or supersoluble groups,
were described in~\cite{DBLP:journals/jalc/CartonPS09,Eilenberg1976,DBLP:conf/icalp/Therien79}.
Investigation of the state complexity of common operations on permutation automata
was initiated on last years edition of this conference~\cite{DBLP:conf/dlt/HospodarM20}.



Here, we investigate the projection operator on permutation automata.
We give a better tight bound for permutation automata, also parameterized
by the number of unobservable transitions, that, however, also grows exponentially.
We give sufficient conditions, related to normal subgroups, to yield
a state-partition permutation automaton for a given projection.
Then, we investigate projections for commuting letters, this in particular
encompasses commutative languages and automata.
We show that if we delete commuting letters by a projection operator, then
we also just need $n$ states for an $n$-state input automaton for the projected language.
In particular this applies
to commutative automata.
We find commutative automata that are not state-partition automata
for a given projection.
This is in particular interesting, as in~\cite{DBLP:journals/tcs/JiraskovaM12}, it was noted that up to then, 
only state-partition automata and automata describing finite language with a unary projected
language were known to have the property that we only need $n$ states for the projected
languages.

Lastly, we derive that the projection operator preserves every variety of commutative languages.
This includes, for example, the commutative 
aperiodic, the commutative group languages or the commutative piecewise-testable languages.


%




\section{General Notions}
\label{subsec:general_notions}

By $\Sigma$ we denote a finite set of symbols, also called an \emph{alphabet}.
By $\Sigma^*$ we denote the set of all \emph{words} over $\Sigma$, i.e., finite sequences
with the concatenation operation. The \emph{empty word} is denoted by $\varepsilon$.
A \emph{language} $L$ is a subset $L \subseteq \Sigma^*$.
Languages using only a single symbol are called \emph{unary languages}.

If $X$ is a set, by $\mathcal P(X) = \{ Y \mid Y \subseteq X \}$ we denote the \emph{power set} of $X$.

If $x$ is a non-negative real number, by $\lceil x \rceil$ we denote the smallest natural number greater or equal to $x$
and by $\lfloor x \rfloor$ the largest natural number smaller or equal to $x$.

Let $\Gamma \subseteq \Sigma$. The homomorphism $\pi_{\Gamma} : \Sigma^* \to \Gamma^*$
given by $\pi_{\Gamma}(x) = x$ for $x \in \Gamma$
and $\pi_{\Gamma}(x) = \varepsilon$ for $x \in \Sigma \setminus \Gamma$
is called a \emph{projection (for $\Gamma$)}.
If $p,q \in Q$, $x \in \Sigma$, then
a transition $\delta(p, x) = q$ is said to be \emph{unobservable}
with respect to the projection $\pi_{\Gamma}$ if $x \in \Sigma \setminus \Gamma$, i.e., $\pi_{\Gamma}(x) = \varepsilon$.
Here, only non-loop 
unobservable transitions are of interest, i.e., those such that $p \ne q$.

A \emph{(partial) deterministic finite automaton (DFA)}
is denoted by a quintuple $\mathcal A =(Q, \Sigma, \delta, q_0, F)$,
where $Q$ is a \emph{finite set of states}, $\Sigma$ the \emph{input alphabet},
$\delta : Q \times \Sigma \to Q$ is a \emph{partial transition function},
$q_0$ the
\emph{start state} and $F \subseteq Q$ the set of \emph{final states}.
The DFA is said to be \emph{complete} if $\delta$ is a total function.
In the usual way, the transition function $\delta$
can be extended to a function $\hat \delta : Q \times \Sigma^* \to Q$ by setting, for $q \in Q$, $u \in \Sigma^*$
and $a \in \Sigma$, $\hat \delta(q, \varepsilon) = q$
and $\hat \delta(q, ua) = \delta(\hat \delta(q, u), a)$.
In the following, we drop the distinction between $\delta$ and $\hat \delta$
and denote both functions simply by $\delta$.

For $S \subseteq Q$ and $u \in \Sigma^*$, 
we set $\delta(S, u) = \{ \delta(s, u) \mid s \in S \mbox{ and } \delta(s, u) \mbox{ is defined} \}$.

The language \emph{recognized} by $\mathcal A$
is $L(\mathcal A) = \{ u \in \Sigma^* \mid \delta(q_0, u) \in F\}$.
A language $L \subseteq \Sigma^*$ is called \emph{regular}, if there exists an automaton $\mathcal A$
such that $L = L(\mathcal A)$.

For $u \in \Sigma^*$, we write $\delta(p, u) = \delta(q, u)$
if both are defined and the results are equal or both are undefined.

We say that $q$ is \emph{reachable} from $p$ (in $\mathcal A$)
if there exists a word $u \in \Sigma^*$ such that $\delta(p, u) = q$.
The DFA $\mathcal A$ is called \emph{initially connected}, if every state is reachable
from the start state. 

The DFA $\mathcal A = (Q, \Sigma, \delta, q_0, F)$
is called \emph{commutative}, if, for each $a,b \in \Sigma$
and $q \in Q$, we have $\delta(q, ab) = \delta(q, ba)$.

Let $\mathcal A = (Q, \Sigma, \delta, q_0, F)$ be a complete DFA.
For a word $u \in \Sigma^*$, the 
\emph{transition function (in $\mathcal A)$ associated to $u$}
is the function $\delta_u : Q \to Q$ 
given by $\delta_u(q) = \delta(q, u)$ for $q \in Q$.
The \emph{transformation monoid} is $\mathcal T_{\mathcal A} = \{ \delta_u \mid u \in \Sigma^* \}$.
Note that we defined the transformation monoid only for complete DFAs, as this is the only context
where we need this notion here.

To denote transitions in permutation DFAs, we use
a \emph{cycle notation} also used in~\cite{DBLP:journals/tcs/BrzozowskiS19,DBLP:conf/dlt/HospodarM20}.
More formally, $(q_1, \ldots, q_k)$ denotes the cyclic permutation
mapping $q_i$ to $q_{i+1}$ for $i \in \{1,\ldots,k-1\}$ and $q_k$ to $q_1$.
For example, $a = (1,2)(3,4,5)$ means the letter $a$ swaps the states $1$ and $2$,
cyclically permutes the states $3,4$ and $5$ in the indicated order and fixes all other states.

A \emph{variety (of formal languages)} $\mathcal V$~\cite{Eilenberg1976,Pin86,DBLP:reference/hfl/Pin97} associates,
to each alphabet $\Sigma$,
a class of recognizable
languages $\mathcal V(\Sigma^*)$ over $\Sigma$ such
that (1) $\mathcal V(\Sigma^*)$ is a boolean algebra,
(2) if $\varphi : \Sigma^* \to \Gamma^*$ is a homomorphism,
then $L \in \mathcal V(\Gamma^*)$
implies $\varphi^{-1}(L) \in \mathcal V(\Sigma^*)$
and (3) if $L \in \mathcal V(\Sigma^*)$ and $x \in \Sigma$,
then $\{ u \in \Sigma^* \mid xu \in L \}$
and $\{ u\in \Sigma^* \mid ux \in L \}$
are in $\mathcal V(\Sigma^*)$.

\begin{toappendix}

 \begin{remark}
 \label{rem:minimal_aut}
 Let $\mathcal A = (Q, \Sigma, \delta, q_0, F)$. 
 Two states $p, q \in Q$
 are said to be \emph{distinguishable},
 if $\delta(p,u)$ and $\delta(q, u)$ 
 are not equal and there exists a word $u \in \Sigma^*$
 such that precisely one of the two states
 $\delta(p, u)$ or $\delta(q, u)$ is defined
 and final.
 A state is said to be \emph{coaccessible}, if a final state
 is reachable from it.
 The DFA $\mathcal A$ is the smallest DFA recognizing
 $L(\mathcal A)$, i.e., has the least number of states,
 if and only if every state is reachable from the start
 state, every two distinct states are distinguishable
 and every state is coaccessible~\cite{HopUll79}.
 We use this characterization
 later in proving our lower bound in Theorem~\ref{thm:lower_bound}.
 Note that if in a DFA $\mathcal A$
 every state is coaccessible,
 then for every non-empty subset $S \subseteq Q$ of states, we find
 a word $u \in \Sigma^*$ such that $\delta(S, u) \cap F \ne \emptyset$.
 This in particular implies that for the projection automaton, as its states are non-empty
 subsets of states, if we assume
 the input permutation automaton is initially connected, which also implies
 that every state is coaccessible for permutation automata\footnote{By Lemma~\ref{lem:inverses}, for every state of the form $q = \delta(q_0, u)$,
 there exists $u'$ such that $\delta(q_0, uu') = q_0$,
 and as $L(\mathcal A) \ne \emptyset$, there exists $w$
 such that $\delta(q_0, w) \in F$. Hence $\delta(q, u'w) \in F$
 and so, all states reachable from the start state are coaccessible.} 
 with $L(\mathcal A) \ne \emptyset$, 
 that 
 every state is coaccessible. So, we only
 need to show reachability and distinguishablity 
 in the proof of Theorem~\ref{thm:lower_bound}.
 \end{remark}
 
\end{toappendix}

\section{Orbit Sets, Projected Languages and Permutation Automata}
\label{subsec:orbitsets}



First, we introduce the orbit set of a set of states for a subalphabet. An orbit set collects those states that are reachable
from a given set of states by only using words from a given subalphabet.
This is also called \emph{unobservable reach} in~\cite{DBLP:books/daglib/0034521}.

\begin{definition}
\label{def:orbit}
 Let $\mathcal A = (Q, \Sigma, \delta, q_0, F)$ be a DFA.
 Suppose $\Sigma' \subseteq \Sigma$ and $S \subseteq Q$. The
 \emph{$\Sigma'$-orbit of $S$} is the set
 \[ 
  \Orb_{\Sigma'}(S) = \{ \delta(q, u) \mid \delta(q, u) \mbox{ is defined, } q \in S \mbox{ and } u \in \Sigma'^* \}.
 \]
 Also, for $q \in Q$, we set $\Orb_{\Sigma'}(q) = \Orb_{\Sigma'}(\{q\})$.
\end{definition}

 Let $\mathcal A = (Q, \Sigma, \delta, q_0, F)$ be a DFA
 and $\Gamma \subseteq \Sigma$.
 Set $\Delta = \Sigma \setminus \Gamma$.
 Next, we define the \emph{projection automaton} of $\mathcal A$
 for $\Gamma$ as 
 $
 \mathcal R_{\mathcal A}^{\Gamma} = (\mathcal P(Q), \Gamma, \mu, \Orb_{\Delta}(q_0), E)
 $
 with, for $S \subseteq Q$ and $x \in \Gamma$, the transition function
  \begin{equation}\label{eqn:def_mu}
     \mu(S, x) = \operatorname{Orb}_{\Delta}(\delta(S, x))
 \end{equation}
 and
 $E = \{ T \subseteq Q \mid T \cap F \ne \emptyset \}$.
 In general, $\mathcal R_{\mathcal A}^{\Gamma}$
 is not initially connected. However, non-reachable states could be omitted. Actually, by the definition of the start state and transition function,
 we can restrict the state set to subsets of the form $\Orb_{\Delta}(S)$
 for $\emptyset \ne S \subseteq Q$.
 
 \begin{theoremrep}
  Let $\mathcal A$ be a DFA
  and $\Gamma \subseteq \Sigma$.
  Then, $\pi_{\Gamma}(L(\mathcal A)) = L(\mathcal R_{\mathcal A}^{\Gamma})$.
 \end{theoremrep}
 \begin{proof}
  Set $\Delta = \Sigma \setminus \Gamma$.
  Intuitively, we construct an $\varepsilon$-NFA for $\pi_{\Gamma}(L(\mathcal A))$
  from $\mathcal A$ by replacing all unobservable transitions
  with $\varepsilon$-transitions and determinize it, where
  the $\varepsilon$-closure of a subset of states corresponds
  to an $\Delta$-orbit.
  More formally, note that by Equation~\eqref{eqn:def_mu}
  and Definition~\ref{def:orbit},
  for $R \subseteq Q$, $q \in Q$ and $x \in \Sigma$,
  \[
   q \in \mu(R, x) \Leftrightarrow \exists p \in Q \ \exists y \in \Delta^* : p \in \delta(R, x) \land \delta(p, y) = q.
  \]
  and $q \in \Orb_{\Delta}(\{q_0\})$ if and only if
  there exists $y \in \Delta^*$ such that $\delta(q_0, y) = q$.
  Hence, we find, inductively, for every two states $p, q \in Q$ and $u \in \Gamma^*$, that
  $q \in \mu(\{p\}, u)$ if and only if, for $n = |u|$,
  there exist words $v_1, \ldots, v_n \in \Delta^*$,
  letters $x_1, \ldots x_n \in \Gamma$
  and states $p_0, p_1, \ldots, p_n \in Q$
  such that
  \begin{align*}
      p_0 & = p \\
      p_n & = q \\
      p_{i+1} & = \delta(p_i, x_{i+1} u_{i+1}) \mbox{ for } i \in \{0, \ldots,n-1\} \\
      u & = x_1 \cdots x_n.
  \end{align*}
  So, if $w \in \Gamma^*$, we can deduce that $\mu(\Orb_{\Delta}(\{q_0\}), w) \cap F \ne \emptyset$
  if and only if, for $n = |w|$,
  there exist words $v_0, v_1, \ldots, v_n \in \Delta^*$,
  letters $x_1, \ldots, x_n \in \Gamma$
  and states $p_0, p_1, \ldots, p_n$
  such that
   \begin{align*}
      p_0 & = \delta(q_0, v_0) \\
      p_n & \in F \\ 
      p_{i+1} & = \delta(p_i, u_i x_{i+1} u_{i+1}) \mbox{ for } i \in \{0, \ldots,n-1\} \\
      w & = x_1 \cdots x_n.
  \end{align*}
  Now, note that, $w \in \pi_{\Gamma}(L(\mathcal A))$
  is equivalent to the existence of $n \ge 0$ and a word $u \in L(\mathcal A)$
  that could be factorized into subwords $u_0, \ldots, u_n \in \Delta$
  and letters $x_1, \ldots, x_n \in \Gamma$
  such that $w = \pi_{\Gamma}(u_0 x_1 u_1 \cdots x_n u_n) = x_1 \cdots x_n$.
  This is equivalent to the above equations
  and so we find
  \[
   \mu(\Orb_{\Delta}(\{q_0\}), w) \cap F \ne \emptyset
   \Leftrightarrow \exists w \in \pi_{\Gamma}(L(\mathcal A)),
  \]
  or $L(\mathcal B) = \pi_{\Gamma}(L(\mathcal A))$.~\qed
 \end{proof}
 
 We do not introduce $\varepsilon$-NFAs formally here, but only refer to the literature~\cite{HopUll79}.
 However, we note in passing that, usually, an automaton for a projected language
 of a regular language is constructed by replacing the letters
 to be deleted by $\varepsilon$-transitions
 and then determinizing the resulting $\varepsilon$-NFA~\cite{HopUll79,DBLP:journals/tcs/JiraskovaM12}.
 Our construction is a more direct formulation of these steps, where the orbit
 sets are used in place of the $\varepsilon$-closure computations.

 In~\cite{DBLP:journals/tcs/JiraskovaM12,DBLP:conf/ifipTCS/JiraskovaM12}, an automaton was called a \emph{state-partition automaton}
 with respect to a projection $\pi_{\Gamma}$ (or, for short, a state-partition automaton for $\Gamma$),
 if the states of the resulting automaton from the above procedure, after discarding non-reachable subsets, form a partition of the original state set. 
 Hence, in our terminology, an automaton $\mathcal A$ 
 %
 is a state partition automaton if the reachable
 states of $\mathcal R_{\mathcal A}^{\Gamma}$
 form a partition of the states of $\mathcal A$.

A \emph{permutation automaton} (or \emph{permutation DFA}) is a DFA
$\mathcal A = (Q, \Sigma, \delta, q_0, F)$
such that every letter permutes the state set, i.e., the function $\delta_x : Q \to Q$
given by $\delta_x(q) = \delta(q, x)$ is a permutation, or bijection, of $Q$ for every $x \in \Sigma$. The languages recognized by permutation automata
are called \emph{group languages}. Note that permutation DFAs are complete DFAs.

The \emph{identity transformation (on $Q$)} is the permutation $\operatorname{id} : Q \to Q$
given by $\operatorname{id}(q) = q$ for each $q \in Q$.

 Next, we take a closer look at the orbit sets
 for permutation automata. But first, a general property of permutation
 automata.

\begin{lemmarep}
\label{lem:inverses}
 Let $\mathcal A = (Q, \Sigma, \delta, q_0, F)$ be a permutation automaton
 and $\Sigma' \subseteq \Sigma^*$.
 Then, for every $u \in \Sigma'^*$ there exists $u' \in \Sigma'^*$
 such that $\delta(q, uu') = q$ for each $q \in Q$, i.e,
 the word $uu'$ represents the identity transformation on $Q$.
\end{lemmarep}
\begin{proof} 

 If $u \in \Sigma'^*$, define a function $\delta_u : Q \to Q$ by
 $\delta_u(q) = \delta(q, u)$.
 Then, $\delta_u$ is a permutation, and as $Q$ is finite, there
 exists $n > 0$ such that $\delta_u^n$ is the identity function, i.e.,
 $\delta_u^n(q) = q$ for each $q \in Q$.
 Then, $\delta(q, u^n) = q$ for each $q \in Q$
 and setting $u' = u^{n-1} \in \Sigma'^*$ gives the claim.~\qed
\end{proof}

 With the previous lemma, we can show that the orbit
 sets for permutation automata partition the state set.
 This property is crucial to derive our state complexity bound
 for projection, as it vastly reduces the possible
 subsets that are reachable in $\mathcal R_{\mathcal A}^{\Gamma}$, namely
 only unions of orbit sets.
 
\begin{lemmarep}
\label{lem:orbits-partition}
 Let $\mathcal A = (Q, \Sigma, \delta, q_0, F)$ be a permutation automaton.
 Suppose $\Sigma' \subseteq \Sigma$. Then, the sets $\Orb_{\Sigma'}(q)$, $q \in Q$,
 partition $Q$ and for every $S \subseteq Q$,
 $
  \Orb_{\Sigma'}(S) = \bigcup_{q \in S} \Orb_{\Sigma'}(q).
 $
\end{lemmarep} 
\begin{proof}
  Clearly, we have $q \in \Orb_{\Sigma'}(q)$ for each $q \in Q$. 
  Suppose $q \in \Orb_{\Sigma'}(s) \cap \Orb_{\Sigma'}(t)$.
  Then, we first show $\Orb_{\Sigma'}(s) \subseteq \Orb_{\Sigma'}(t)$.
  Let $p \in \Orb_{\Sigma'}(s)$.
  By the previous assumptions, there exist $u, v, w \in \Sigma'^*$
  such that $q = \delta(s, u) = \delta(t, v)$
  and $p = \delta(s, w)$.
  By Lemma~\ref{lem:inverses}, we can choose $u' \in \Sigma'^*$
  such that $\delta(s, uu') = s$.
  So,
  \begin{align*}
      \delta(t, vu'w) & = \delta(\delta(t, v), u'w) \\ 
                      & = \delta(\delta(s, u), u'w) \\
                      & = \delta(\delta(s, uu'), w) \\
                      & = \delta(s, w) \\
                      & = p.
  \end{align*}
  Hence, $p \in \Orb_{\Sigma'}(t)$.
  As $p$ was chosen arbitrary, we find $\Orb_{\Sigma'}(s) \subseteq \Orb_{\Sigma'}(t)$.
  Similarly, we can show the other inclusion
  and hence $\Orb_{\Sigma'}(s) = \Orb_{\Sigma'}(t)$.

  If $\Orb_{\Sigma'}(q) \cap \Orb_{\Sigma'}(S) \ne \emptyset$,
  then there exists $s \in S$ such that $\Orb_{\Sigma'}(q) \cap \Orb_{\Sigma'}(s) \ne \emptyset$.
  So, by the previous claim, both $\Sigma'$-orbits are equal
  and we find $\Orb_{\Sigma'}(q) \subseteq \Orb_{\Sigma'}(S)$.
  By considering all elements in $S$, this yields $\bigcup_{q \in S} \Orb_{\Sigma'}(q) = \Orb_{\Sigma'}(S)$.~\qed
\end{proof}

%
%

\section{Projection on Permutation Automata}


Here, we state a tight upper bound
for the number of states of the projection
of a language recognized by a permutation automaton.

Our bound is parameterized by the number of states of the input
automaton and by the number of non-loop unobservable
transitions. More specifically, we consider the number of states that are incident
with non-loop unobservable transitions. Hence, we
disregard unobservable multi-transitions and do not take the direction into account, i.e., counting
multiple transitions resulting from multiple letters
between the same states only once and do not take their direction into account.
This is the same usage of this parameter as  in~\cite{DBLP:conf/ifipTCS/JiraskovaM12} for the general
case.

\begin{theorem}
\label{thm:upper_bound}
 Let $\mathcal A = (Q, \Sigma, \delta, q_0, F)$
 be a permutation DFA and $\Gamma \subseteq \Sigma$.
 Set 
 $
  m = |\{\ p, q \in Q \mid p \ne q \mbox{ and } q \in \delta(p, \Sigma \setminus \Gamma) \}|.
 $
 Then, if $m > 0$, the projected language $\pi_{\Gamma}(L(\mathcal A))$
 is recognizable by a DFA
 with at most $2^{\mathstrut{|Q| - \lceil \frac{m}{2} \rceil}} - 1$
 states and if $m = 0$, the projected language is recognizable 
 by a DFA with at most $|Q|$ states.
\end{theorem}
\begin{proof}
 %
 %
 Set $\Delta = \Sigma \setminus \Gamma$,
 $S = \{ p, q \in Q \mid p \ne q \mbox{ and } q \in \delta(p, \Delta) \}$ and $T = \{ p \in Q \mid \forall x \in \Delta : \delta(p, x) = p \}$.
 Then, as $\mathcal A$ is a permutation automaton,
 $Q$ is the disjoint union of $S$ and $T$ and
 \begin{equation}
 \label{eqn:S_and_T}
  q \in T \Leftrightarrow \Orb_{\Delta}(q) = \{q\} \mbox{ and }
  q \in S \Leftrightarrow |\Orb_{\Delta}(q)| \ge 2.
 \end{equation}
 Set $\mathcal B = \mathcal R_{\mathcal A}^{\Gamma}$.
 If $m = 0$, then $Q = T$ and every $a \in \Delta$ 
 induces a self-loop at every state.
 In this case, it is clear that we can simply leave out all the transitions
 labeled with letters from $\Delta$
 and the resulting permutation automaton recognizes $\pi_{\Gamma}(L(\mathcal A))$.
 More formally, in the definition of $\mathcal B$, in this case,
 the starting state is $\{q_0\}$
 and as $\mathcal A$ is deterministic we have $|\delta(R, x)| \le |R|$
 for every $R \subseteq Q$. So, as the empty set is never reachable for permutation DFAs,
 only the singleton sets $\{q\}$
 are reachable in $\mathcal B$.
 
 Now, suppose $m = |S| > 0$, which implies $m \ge 2$.
 
 \begin{claiminproof}
  Let $m > 0$. Then, in $\mathcal B$ at most  $2^{|Q| - \lceil \frac{m}{2} \rceil} - 1$
  states are reachable from the start state.
 \end{claiminproof}
 \begin{claimproof}
  With the assumption $m > 0$, there exists $q \in Q$
  such that $|\Orb_{\Delta}(q)| > 1$.
  By Equation~\eqref{eqn:S_and_T} and Lemma~\ref{lem:orbits-partition},
  we have at most $|T| + \Big\lfloor \frac{|S|}{2} \Big\rfloor$
  many $\Delta$-orbits, where the maximum number of $\Delta$-orbits is reached if
  every $\Delta$-orbit of a state from $S$ has size exactly two if $|S|$
  is even or every such orbit has size two, except one
  that has size three, if $|S|$ is odd.
  By Lemma~\ref{lem:orbits-partition},
  the sets $\Orb_{\Delta}(\{q\})$, $q \in Q$,
  partition the state set and, for every $R \subseteq Q$,
  we have $\bigcup_{q \in R} \Orb_a(q) = \Orb_a(R)$.
  So, by Equation~\eqref{eqn:def_mu}, every set
  reachable is a union of $\Delta$-orbits,
  i.e., every such set corresponds uniquely 
  to a subset of $\Delta$-orbits for a single state.
  Finally, note that, as $\mathcal A$ is a permutation automaton,
  and hence complete, we have $\delta(R, x) \ne \emptyset$
  for every non-empty $R \subseteq Q$, which also gives that in $\mathcal B$
  the empty set is not reachable.
  So, in total, we find that at most 
  \[ 2^{|T| + \Big\lfloor \frac{|S|}{2} \Big\rfloor} - 1
   = 2^{|Q| - m + \lfloor \frac{m}{2} \rfloor} - 1
   = 2^{|Q| - \lceil \frac{m}{2} \rceil} - 1
  \]
  subsets of states are reachable.
 \end{claimproof}
 
  
%

 So, we have shown the upper bound.~\qed 
\end{proof}

\begin{toappendix}
For the lower bound, we need the following well-known fact from permutation group theory~\cite{Isaacs95}.

\begin{lemma}\label{lem:transposition_and_cycle_gen_Sn}
 Let $\sigma, \tau : \{1,2,\ldots, n\} \to \{1,2,\ldots,n\}$ for $n \ge 2$
 be the permutations $\sigma = (1,\ldots,n)$
 and $\pi = (1,2)$.
 Then, every permutation of $\{1,2,\ldots,n\}$ could be written as a product
 of the permutations~$\sigma$ and~$\tau$.
\end{lemma}
\end{toappendix}

Next, we show that the bound stated in the previous theorem
is actually tight for permutation automata.

%
%

 \begin{figure}
\centering
\scalebox{.75}{
\tikzset{elliptic state/.style={draw,ellipse}}
\begin{tikzpicture}[>=latex',shorten >=1pt,node distance=3.2cm and 2.1cm,on grid,auto,baseline=0]
 \node [elliptic state] (1) {$1$};
 \node [elliptic state, right =of 1] (2) {$2$};
 \node [elliptic state, right =of 2] (3) {$3$};
 \node [elliptic state, right =of 3] (4) {$4$};
 \node [right =of 4] (5) {$\ldots$};
 \node [right =of 5] (5b) {$\ldots$};
 \node [elliptic state, right =of 5b] (6) {$2m-1$};
 \node [elliptic state, right =of 6] (7) {$2m$};
 \node [initial below,accepting,elliptic state, below =of 1] (8) {$n$};
 \node [elliptic state, right =of 8] (9) {$n - 1$};
 \node [right =of 9] (10) {$\ldots$};
 \node [elliptic state, right =of 10] (11) {$2m + 3$};
 \node [elliptic state, right =of 11] (12) {$2m + 2$};
 \node [elliptic state, right =of 12] (13) {$2m + 1$};
 
 \path[->] (1) edge [bend left] node {$d, e$} (3);
 \path[->] (2) edge [bend left] node {$d, e$} (4);
 \path[->] (3) edge [bend left] node {$d$} (1);
 \path[->] (4) edge [bend left] node {$d$} (2);
 \path[->] (6) edge [bend right=40,above] node {$e$} (1);
 \path[->] (7) edge [bend right=40,above] node {$e$} (2); 
 \path[->] (3) edge [bend left] node {$e$} (5);
 \path[->] (4) edge [bend left] node {$e$} (5b);
 \path[->] (5) edge [bend left] node {$e$} (6);
 \path[->] (5b) edge [bend left] node {$e$} (7);

 \path[->] (1) edge [bend left, right] node {$f,g$} (8);
 \path[->] (8) edge [bend left, right] node {$f,g$} (1);
 \path[->] (2) edge [bend left] node {$g$} (9);
 \path[->] (9) edge [bend left,right] node {$g$} (2);  
 
 \path[->] (1) edge [bend left,below] node {$a$} (2);
 \path[->] (2) edge [bend left,above] node {$a$} (1);
  \path[->] (3) edge [bend left,below] node {$a$} (4);
 \path[->] (4) edge [bend left,above] node {$a$} (3);
 \path[->] (6) edge [bend left,below] node {$a$} (7);
 \path[->] (7) edge [bend left,above] node {$a$} (6);
 
 \path[->] (12) edge [bend right=15] node {$b$} (13);
 \path[->] (13) edge [bend right,above] node {$b,c$} (12);
 
 \path[->] (12) edge [bend right,above] node {$c$} (11);
 \path[->] (11) edge [bend right,above] node {$c$} (10);
 \path[->] (10) edge [bend right,above] node {$c$} (9);
 \path[->] (9)  edge [bend right,above] node {$c$} (8);
 \path[->] (8)  edge [bend right=15] node {$c$} (13);
  
\end{tikzpicture}}
\caption{All transitions not shown, for example for the letter $b$ at the state $n$,
 correspond to self-loops, as permutation automata are complete. 
 Then, the permutation automaton shown reaches the upper bound stated in Theorem~\ref{thm:upper_bound}
 for the projection $\pi_{\Gamma} : \{a,b,c,d,e,f,g\}^* \to \Gamma^*$ with $\Gamma = \{b,c,d,e,f,g\}$. }
 \label{fig:lower_bound}
\end{figure}
 
\begin{theoremrep}
\label{thm:lower_bound}
 Let $n, m > 0$ be such that $0 < 2m + 1 < n$, $\Sigma = \{a,b,c,d,e,f,g\}$
 and $\Gamma = \{b,c,d,e,f,g\}$.
 Then, there exists a permutation automaton $\mathcal A = (Q, \Sigma, \delta, q_0, F)$
 with $2m$
 states incident to non-loop unobservable transitions for~$\pi_{\Gamma}$, i.e.,
 \[
  2m = |\{\ p, q \in Q \mid p \ne q \mbox{ and } q \in \delta(p, \Sigma \setminus \Gamma) \}|,
 \]
 such that every DFA for $\pi_{\Gamma}(L(\mathcal A))$
 needs at least $2^{n - m} - 1$ 
 states.
\end{theoremrep}
\begin{proofsketch}
 See Figure~\ref{fig:lower_bound} for a permutation automaton
 giving the lower bound.
 The automaton has $n$ states, and the letters act the following way:
 \[ 
 \begin{array}{llll}
     a & = (1,2)(3,4) \cdots (2m-1, 2m), \\ 
     b & = (2m + 1, 2m + 2), & c = (2m + 1, 2m + 2, \ldots, n), \\
     d & = (1,3)(2,4), & e = (1,3,\ldots,2m-1)(2,4,\ldots, 2m), \\
     f & = (1,n), & g = (1,n)(2, n-1).
 \end{array}
 \]
 With $\Delta = \{a\}$, the $\Delta$-orbits are 
 $\{1,2\},\{3,4\},\ldots, \{2m-1,2m\}, \{2m+1\}, \ldots, \{n\}$.
 The letters $b,c$ are chosen such that every permutation of the 
 states $\{2m+1,\ldots,n\}$ could be written as a word over them,
 and the letters $d$ and $e$ such that every permutation on the
 $\Delta$-orbits $\{1,2\}, \ldots, \{2m-1,2m\}$ could be written
 as a word over them. The letters $f$ and $g$ help to map between
 these $\Delta$-orbits in such a way that every non-empty
 union of $\Delta$-orbits is reachable, and all these $\Delta$-orbits
 give distinguishable states. By mapping onto the two element $\Delta$-orbits
 and back, we can enlarge the sets that are reachable.\qed

\end{proofsketch}
\begin{proof}
 Let $\Sigma = \{a,b,c,d,e,f,g\}$ and $\mathcal A = (\Sigma, Q, \delta, q_0, F)$
 with $Q = \{1,2,\ldots, n\}$,
 start state $q_0 = n$ and $F = \{ n \}$. 
 The transitions were already defined in the proof sketch
 in the main text. So, here we only give
 the proof in detail.
 Note that, the letters $d$ and $e$
 permute the subsets $\{ i, i + 1\}$
 for $i \in \{1,\ldots, 2m\}$ and we
 have, for each $i, j \in \{1,3,\ldots, 2m-1\}$, 
 that there exists, by Lemma~\ref{lem:transposition_and_cycle_gen_Sn}
 and the definition of the transition function in the projection automaton,
 a word $u \in \{d,e\}^*$ 
 such that 
 \[
  \mu(\{ i, i + 1 \}, u) = \{ j, j + 1 \}.
 \]
 Also, observe that for $u \in \{b,c\}^*$ and $v \in \{d,e\}^*$
 we have $\delta(q, uv) = \delta(q, vu)$ for each $q \in Q$.

 Set $\Delta = \{a\}$ and $\Gamma = \{b,c,d,e,f,g\}$. We will refer to the projection automaton
 \[
  \mathcal R_{\mathcal A}^{\Gamma} = (\mathcal P(Q), \Gamma, \mu, \Orb_{\Delta}(q_0), E) 
 \]
 in the following. Here, we have $\Orb_{\Delta}(q_0) = \{q_0\}$.

 Let $S$
 be the set
 \begin{multline*}
     \{ \emptyset \ne A \subseteq Q \mid 
     \forall i \in \{1,3,\ldots, 2m - 1\} : \{ i, i + 1 \} \subseteq A \mbox{ or } \{ i, i + 1 \} \cap A = \emptyset \}.
 \end{multline*}
 Then, $S = \{ \Orb_{\Delta}(A) \mid \emptyset \ne A \subseteq Q \}$
 and $S$ could be seen as a subset of states for the projection 
 automaton $\mathcal R_{\mathcal A}^{\Gamma}$.

 The following equations are readily implied and will be used frequently without further mention:
 \begin{align}
     \delta(A, u)  & = \mu(A, u) \mbox{ for } A \in S, u \in \{d,e\}^* \mbox{ or } u \in \{b,c\}^*; \label{eqn:b_c_d_e_permute_a_orbits} \\
     \mu(\{n\}, x) & = \{1,2\} \mbox{ for } x \in \{f,g\}; \label{eqn:g_and_n}\\ 
     \mu(\{1,2\}, f) & = \{1,2\} \cup \{n\}; \\ 
     \mu(A \cup B, u) & = \mu(A, u) \cup \mu(B, u) \mbox{ for } A, B \subseteq Q, u \in \Gamma^*.
 \end{align}
 The last one also implies, for $A \subseteq B$,
 that $\mu(A, u) \subseteq \mu(B, u)$ for $u \in \Gamma^*$.

 We will also use the following function $h : \{1,\ldots, n\} \to \{1,\ldots, n - m \}$
 given by $h(i) = h(i+1) = \frac{i+1}{2}$ for $i \in \{1,3,\ldots, 2m - 1 \}$
 and $h(i) = i - m$ for $i \in \{ 2m + 1, \ldots,n\}$
 in the following arguments.

 \begin{claiminproof}
  We have $|S| = 2^{n - m} - 1$. 
 \end{claiminproof}
 \begin{claimproof}
  We have $S = \{ h^{-1}(A) \mid \emptyset \ne A \subseteq \{1,\ldots, n - m\} \}$
  and $h^{-1}(A) \ne h^{-1}(B)$ for distinct $A, B \subseteq \{1,\ldots, n -m \}$.
  So, as there exist $2^{n - m} - 1$ non-empty subsets
  of $\{1,\ldots, n - m\}$, the claim follows.
 \end{claimproof}
 
 \begin{claiminproof}
  The subsets in $S$ are all reachable as states of $\mathcal R_{\mathcal A}^{\Gamma}$.
 \end{claiminproof}
 \begin{claimproof}
  If $A \in S$, then there exists a non-empty set $B \subseteq \{1,\ldots,n-m\}$
  such that $A = h^{-1}(B)$.
  We do induction on the size of $B$.
  
  \begin{enumerate}
  \item Suppose $|B| = 1$. 
   First, if $h^{-1}(B) \subseteq \{2m+1,\ldots,n\}$,
   then $|h^{-1}(B)| = 1$.
   There exists a permutation of $\{2m+1,\ldots,n\}$
   that maps $q_0$ to the state in $h^{-1}(B)$
   Then, by Lemma~\ref{lem:transposition_and_cycle_gen_Sn},
   there exists $u \in \{b,c\}^*$
   such that $\delta(q_0, u) \in h^{-1}(B)$.
   As the letters $b$ and $c$
   permute the $\Delta$-orbits, see also Equation~\eqref{eqn:b_c_d_e_permute_a_orbits}, 
   and do not map anything in $\{2m+1,\ldots,n\}$
   to a state outside of this set,
   we have, for every prefix\footnote{The word $v$ is a prefix of $u$, if there exists $w \in \Sigma^*$ such that $u = vw$.} $v$ of $u$,
   \[
   \mu(\{q_0\}, v) = \Orb_{\Delta}(\delta(\{ q_0 \}, v) = \{ \delta(q_0, v) \} 
   \] 
   and so $\mu(\{q_0\}, u) = \{ \delta(q_0, u) \} = h^{-1}(B)$.
  
   Otherwise, suppose $h^{-1}(B) = \{ i, i + 1 \}$
   for $i \in \{1,3,\ldots, 2m-1\}$.
   We have
   \[
    \mu(\{q_0\}, f) = \Orb_{\Delta}(\delta(\{q_0\}, f)) = \{1,2\}. 
   \]
   By Lemma~\ref{lem:transposition_and_cycle_gen_Sn},
   there exists a word $u \in \{d,e\}^*$ such that
   $\delta(\{1,2\}, u) = \{i,i+1\}$
   and as the letters $d$ and $e$
   permute the $\Delta$-orbits, see also Equation~\eqref{eqn:b_c_d_e_permute_a_orbits}, 
   and do not map anything in $\{1,\ldots,2m\}$
   to a state outside of this set,
   we find $\mu(\{1,2\}, u) = \{i, i + 1\}$
   and so $\mu(\{q_0\}, fu) = h^{-1}(B)$.
   
   \medskip 
   
  \item Suppose $|B| > 1$.
   Let $k \in B$ and consider $C = B \setminus \{k\}$.
   Inductively, there exists $u \in \Gamma^*$
   such that $\mu(\{q_0\}, u) = h^{-1}(C)$. So, we only need to show that we could reach $h^{-1}(B)$
   from $h^{-1}(C)$.
   We distinguish several cases.
   
   \begin{enumerate}
   \item Suppose $|h^{-1}(k)| = 2$ and $h^{-1}(C) \cap \{2m+1, \ldots, n\}\ne \emptyset$.   
    
    \medskip
    
    Let $v \in \{b,c\}^*$
    describe the permutation that interchanges $n$
    with some state in $h^{-1}(C)$
    (or the identity transformation if $n \in h^{-1}(C)$) and maps everything else to itself
    and, similarly, $w \in \{d,e\}^*$ the permutation
    that interchanges $h^{-1}(k)$ with $\{1,2\}$ (or the identity transformation if $h^{-1}(k) = \{1,2\}$)
    and maps everything else (also the other $\Delta$-orbits)
    to itself. Note that $\delta_{vv}$ and $\delta_{ww}$ 
    represent the identity transformation.
    Existence of these words is guaranteed by Lemma~\ref{lem:transposition_and_cycle_gen_Sn}.
    Then,
    \begin{equation}\label{eqn:no_12_in_set}
     \{1,2\} \cap \mu(h^{-1}(C), w) = \emptyset.
    \end{equation}
    and 
    \[ 
     n \in \mu(h^{-1}(C), v). 
    \]
    Also, note that $\delta(n, v) \in h^{-1}(C)$, so $w$ acts like the identity transformation on $\{2m+1,\ldots,n\}$
    and $\{n\} \cap h^{-1}(C) \setminus\{\delta(n,v)\} = \emptyset$.
    Recall that the words $v$ and $u$ permute the $\Delta$-orbits,
    and so application of $\delta$ to $\Delta$-orbits and $\mu$ give the same result,
    see also Equation~\eqref{eqn:b_c_d_e_permute_a_orbits}.
    Using all this, we find 
    \begin{align*} 
    &  \mu(h^{-1}(C), vwffvw) \\
     & = \mu(\mu(h^{-1}(C), v), wffvw) \\ 
      & = \mu(\{n\} \cup \mu(h^{-1}(C) \setminus \{\delta(n,v)\}, wffvw) \\
      & =\mu(\{n\}, wffvw) \cup \mu(h^{-1}(C) \setminus\{\delta(n,v)\}, wffvw) \\
        & = \mu(\{n\}, ffvw) \cup \mu(h^{-1}(C) \setminus\{\delta(n,v)\}, wvw) & \mbox{[Equation~\eqref{eqn:no_12_in_set}]} \\
       & = \mu(\{1,2\}, fvw) \cup \mu(h^{-1}(C) \setminus \{\delta(n,v)\}, wvw) \\
       & = \mu(\{1,2,n\}, vw) \cup \mu(h^{-1}(C)\setminus \{\delta(n,v)\}, wvw) \\
       & = \mu(\{1,2\},vw) \cup \{ \delta(n, vw) \} \cup \mu(h^{-1}(C) \setminus \{\delta(n,v)\}, wvw) \\
       & =\mu(\{1,2\},vw) \cup \{ \delta(n, vw) \} \cup \mu(h^{-1}(C) \setminus \{n\},vww) \\ 
       & = \mu(\{1,2\},vw) \cup \{ \delta(n, v) \} \cup \mu(h^{-1}(C) \setminus \{n\},ww) \\ 
       & = \mu(\{1,2\},vw) \cup \{ \delta(n, v) \} \cup \mu(h^{-1}(C) \setminus \{n\},ww) \\ 
       & = \mu(\{1,2\},w) \cup \{ \delta(n, v) \} \cup \mu(h^{-1}(C) \setminus \{n\},ww) \\
       & = \mu(\{1,2\},w) \cup \{ \delta(n, v) \} \cup \mu(h^{-1}(C) \setminus \{n\},ww) \\
        & = \mu(\{1,2\},w) \cup \{ \delta(n, v) \} \cup h^{-1}(C) \setminus \{n\} \\
       & = h^{-1}(k) \cup h^{-1}(C) \\
       & = h^{-1}(B).
    \end{align*}
    
     \item Suppose $|h^{-1}(k)| = 2$ and $h^{-1}(C) \cap \{1, \ldots, 2m\}\ne \emptyset$.   
    
    \medskip
    
     In that case, there exists some minimal $i \in \{1,3,\ldots, 2m-1\}$
     such that $\{i,i+1\} \subseteq h^{-1}(C)$.
     We only sketch the argument, as it is similar to the other cases, and give a more bird's eye view here. Existence
     of all the words is guaranteed by Lemma~\ref{lem:transposition_and_cycle_gen_Sn},
     and, as before, we choose some of the words such that $\delta$ and $\mu$
     coincide on the $\Delta$-orbits (but, of course, not on the final concatenation of all words, as
     we want to enlarge the sets).
     First, choose a word $v \in \{d,e\}^*$ that 
     interchanges $\{1,2\}$ and $\{i,i+1\}$ and maps every other $\Delta$-orbit
     of a single state to itself, or let $v$ be the identity transformation
     if $\{1,2\} = \{i,i+1\}$.
     Then, apply $v$ to $h^{-1}(C)$ and then apply $g$.
     The resulting set contains $\{n,n-1\}$.
     Then, apply~$f$. The result contains $\{1,2\}, \{n-1\}$ and all the previous 
     states. 
     Then, apply $v$ again. This, essentially, reverses all the previous operations 
     on the states in $\{1,\ldots,2m\}$. So, we get the state set $h^{-1}(C) \cup \{n-1\}$.
     Let $w \in \{d,e\}^*$ be a word that interchanges 
     $\{1,2\}$ with $h^{-1}(k)$ or the identity transformation if both are equal.
     Then, if we apply $w$ we end up with a set in $\{1,\ldots, 2m\}$ that does not contain $\{1,2\}$
     but $\{n-1\}$. Then, apply $g$
     and $\{n-1\}$ gets mapped to $\{1,2\}$ by $\mu$, and everything else remains the same.
     A final application of $w$ gives $h^{-1}(k) \cup h^{-1}(C)$.
    
     \medskip
    
    \item $|h^{-1}(k)| = 1$ and $h^{-1}(C) \cap \{2m+1,\ldots,n\} \ne \emptyset$.
    
     \medskip 
     
     By Lemma~\ref{lem:transposition_and_cycle_gen_Sn}
     there exist $v, w \in \{b,c\}^*$
     such that
     \begin{enumerate}
     \item $v$ represents the permutation that interchanges $n-1$
      with the state in $h^{-1}(k)$ or  the identity permutation
      if $h^{-1}(k) = \{n-1\}$,
     
     \item $w$ represents the permutation that interchanges $n$
      with some state in $\delta(h^{-1}(C) \cap \{2m+1,\ldots,n\}, v) \subseteq \{2m+1,\ldots,n\}$
      or the identity permutation if $n \in \delta(h^{-1}(C) \cap \{2m+1,\ldots,n\}, v)$.
     \end{enumerate}
     
     Then, also using $h^{-1}(k) \subseteq \{ 2m+1, \ldots,n \} \setminus h^{-1}(C)$
     and Equation~\eqref{eqn:b_c_d_e_permute_a_orbits},
     we have $n - 1 \notin \mu(h^{-1}(C), v)$
     in both cases $n - 1 \in h^{-1}(C)$, when the state in $h^{-1}(k)$
     gets mapped to $n - 1$, or $n - 1\notin h^{-1}(C)$,
     as then the states $\{ h^{-1}(k), n - 1\}$ outside of $h^{-1}(C)$
     are swapped (or mapped to itself if they are equal) and everything else is mapped to itself.
     Furthermore,
     \[
     n - 1 \notin \mu(h^{-1}(C), vw)
     \] 
     as $n - 1 \notin \mu(h^{-1}(C), v)$
     implies that $w$ interchanges $n$ with some state different
     from $n - 1$ or is the identity permutation.
     Also, by choice of $w$,
     \[
       n  \in  \mu(h^{-1}(C), vw).
     \]
     So, by Equation~\eqref{eqn:g_and_n},
     \[ 
      \{1,2\} \subseteq \mu(h^{-1}(C), vwg) 
     \]
     and hence
     \[
      \{n-1,n\} \subseteq \mu(h^{-1}(C), vwgg). 
     \]
     As $g$ only moves $\{1,2,n-1,n\}$ by swapping $1$ with $n$ and $2$ with $n-1$,
     and by the choice of $w$, we have
     \[
      \mu(h^{-1}(C), vwggw) = \mu(h^{-1}(C), v) \cup \mu(\{n-1\}, w).
     \]
     As written above, because of $n - 1 \notin \mu(h^{-1}(C), v)$,
     we have $\delta(n-1, w) = n -1$ and so, by Equation~\eqref{eqn:b_c_d_e_permute_a_orbits}, 
     the above equals
     \[
      \mu(h^{-1}(C), v) \cup \{n-1\}.
     \]
     By choice of $v$ and Equation~\eqref{eqn:b_c_d_e_permute_a_orbits}, 
     $\mu(h^{-1}(C), vv) = h^{-1}(C)$ 
     and $\mu(\{n-1\}, v) = h^{-1}(k)$.
     So,
     \[
      \mu(h^{-1}(C), vwggwv) = h^{-1}(C) \cup h^{-1}(k) = h^{-1}(B).
     \]
     
    \item $|h^{-1}(k)| = 1$ and $h^{-1}(C) \cap \{2m+1,\ldots,n\} = \emptyset$.   
      
      \medskip 
      
      Then, as $C \ne \emptyset$,
      there exists $i \in \{1,\ldots, 2m\}$
      such that $\{i,i+1\} \subseteq h^{-1}(C)$.
      Let $v \in \{d,e\}^*$ be the permutation that
      interchanges $\{i, i +1\}$ with $\{1,2\}$
      or the identity transformation if $\{i,i+1\} = \{1,2\}$.
      Also, let $w \in \{b,c\}^*$ be a
      permutation that maps $n$ to the state in $h^{-1}(k)$.
      Then,
      \[
       n \in \mu(h^{-1}(C), vf).
      \]
      Note that, after applying $vf$, as $2$ is not moved by $f$, the $\Delta$-orbit
      $\{1,2\}$ is still contained in $\mu(h^{-1}(C), vf)$.
      Then,
      \[
       \mu(h^{-1}(C), vfvw) = h^{-1}(C) \cup h^{-1}(k) = h^{-1}(B).
      \]
    \end{enumerate}

   So, we have shown that $h^{-1}(B)$ is reachable from $\{q_0\}$
   in $\mathcal R_{\mathcal A}^{\Gamma}$ in all cases for $h^{-1}(C)$.
  \end{enumerate}

   By induction, we can conclude that every set $h^{-1}(B) \in S$
   with $B \ne \emptyset$ is reachable in $\mathcal R_{\mathcal A}^{\Gamma}$.
 \end{claimproof}

 \begin{claiminproof}
  The subsets in $S$ are distinguishable as states of $\mathcal R_{\mathcal A}^{\Gamma}$.
 \end{claiminproof}
 \begin{claimproof}
  Let $A, B \in S$ be distinct.
  Without loss of generality, suppose $q \in A \setminus B$.
  We distinguish two cases.
  \begin{enumerate}
  
  \item Suppose $q \in \{2m+1,\ldots,n\}$. By\footnote{The letter $b$ interchanges $2m+1$
  and $2m+2$ and the letter $c$ cyclically permutes the states $\{2m+1,\ldots,n\}$.
  So, by Lemma~\ref{lem:transposition_and_cycle_gen_Sn},
  every permutation of the states $\{2m+1,\ldots,n\}$ could be written
  as a word over $b$ and $c$.}
  Lemma~\ref{lem:transposition_and_cycle_gen_Sn},
  there exists $u \in \{b,c\}^*$
  such that $\delta(q, u) \in F$.
  As $|F| = 1$ and $u$ is a permutation of all the states $Q$,
  for each $q' \in Q \setminus \{q\}$,
  we have $\delta(q', u) \notin F$.
  In particular, we find $\delta(B, u) \cap F = \emptyset$
  and $\delta(A, u) \cap F \ne \emptyset$.
  Then, note that for $q \in \{2m+1,\ldots,n\}$,
  and $u \in \{b,c\}^*$, as $\Orb_{\Delta}(q)$ is a singleton set
  and $\delta(\{2m+1, \ldots, n), b) = \delta(\{2m+1,\ldots,n\}, c) = \{2m+1,\ldots,n\}$,
  we find $\mu(A, u) = \delta(A, u)$
  and $\mu(B, u) = \delta(B, u)$. 
  Hence, $\mu(B, u) \cap F = \emptyset$
  and $\mu(A, u) \cap F \ne \emptyset$
  and $A$ and $B$ are distinguishable in $\mathcal R_{\mathcal A}^{\Gamma}$.
  \item Suppose $q \in \{1,\ldots,2m\}$. Then
   $\Orb_{\Delta}(q) \subseteq A$.
   By Lemma~\ref{lem:transposition_and_cycle_gen_Sn}, there
   exists $u \in \{d,e\}^*$
   such that $\delta(\Orb_{\Delta}(q), u) = \{1,2\}$.
   As the letters $d$ and $e$
   permute the $\Delta$-orbits, for every prefix
   $v$
   of $u$, we have
   \[
    \mu(\Orb_{\Delta}(q), v) = \Orb_{\Delta}(p)
   \]
   for some $p \in \{1,\ldots,2m\}$, i.e, the image is the $\Delta$-orbit of a single state,
   and so
   $\mu(\Orb_{\Delta}(q), u) = \{1,2\}$.
   Then,
   \[
    \mu(\{1,2\}, f) = \Orb_{\Delta}(\delta(\{1,2\}, f))
     = \Orb_{\Delta}(\{2,n\}) = \{1,2,n\}
   \]
   and so $\mu(A, uf) \cap F \ne \emptyset$.
   As $uf$ permutes the states $Q$,
   we have $\delta(B, uf) \cap F = \emptyset$.
   Also, as $u$ permutes the $\Delta$-orbits,
   \[
    \mu(\Orb_{\Delta}(q), u) \cap \{1,2\} = \emptyset.
   \]
   So $\mu(\Orb_{\Delta}(q), uf) \cap F = \emptyset$ (more specifically, here $\mu(\Orb_{\Delta}(q), uf) = \mu(\Orb_{\Delta}(q), u)$).
   \end{enumerate}
  Hence, in both cases the states $A$
  and $B$ are distinguishable in $\mathcal R_{\mathcal A}^{\Gamma}$.
 \end{claimproof}

 So, the states in $S$ of $\mathcal R_{\mathcal A}^{\Gamma}$
 are all reachable and distinguishable. 
 Also, the input automaton is initially connected,
 so, as written at the end of Remark~\ref{rem:minimal_aut},
 every state is coaccessible, and so, also as written in Remark~\ref{rem:minimal_aut},
 every state in $\mathcal R_{\mathcal A}^{\Gamma}$
 is coaccessible. Note that we needed every state to be coaccessible, as
 we are talking about minimal DFAs, which could be partial. For example, a reachable non-coaccessible
 non-final state might be distinguishable from every other state, but nevertheless could be left out.
 Hence, the minimal
 automaton needs at least $|S|$ states and the statement follows.~\qed
\end{proof}

\begin{remark}
 Note that if $\mathcal A = (Q, \Sigma, \delta, q_0, F)$ is initially connected
 and
 has the property that from every state $q \in Q$ a final state
 is reachable, then also $\mathcal R_{\mathcal A}^{\Gamma}$
 has this property. As permutation DFAs are complete by definition,
 this implies that the tight bound stated in Theorem~\ref{thm:upper_bound}
 remains the same if we would additionally demand the resulting DFA for the projection
 to be complete. 
\end{remark}

We used an alphabet of size seven to match the bound. So, the question arises
if we can reach the bound using a smaller alphabet. I do not know the answer yet, but 
by using a result from~\cite[Theorem 6]{DBLP:journals/tcs/JiraskovaM12}
that every projection onto a unary language needs less than $\exp((1 + o(1))\sqrt{n \ln(n)})$
states, we can deduce that we need at least a ternary alphabet to reach the bound
stated in Theorem~\ref{thm:upper_bound}.
For the bound stated in Theorem~\ref{thm:upper_bound} is lowest possible, apart from the trivial
case $m = 0$, 
if $m = n$. Then, the bound is $2^{\lfloor n / 2 \rfloor} - 1$.
However, asymptotically, this grows way faster than $\exp((1 + o(1))\sqrt{n \ln(n)})$.
In fact, the ratio of both expressions could be arbitrarily large.

\begin{proposition}
 Each permutation automaton $\mathcal A$ such that $\pi_{\Gamma}(L(\mathcal A))$, for a non-empty and proper subalphabet $\Gamma \subseteq \Sigma$,
 attains the bound
 stated in Theorem~\ref{thm:upper_bound} with $m > 0$ must be over an alphabet with at least three letters
 and $|\Gamma| \ge 2$.
\end{proposition}

\section{State-Partition Automata and Normal Subgroups}
\label{sec:normal_subgroups}

First, we derive a sufficient condition for a permutation automaton
to be a state-partition automaton for a projection. Then, we introduce normal subgroups and show that if the letters
generate a normal subgroup, this condition is fulfilled.

\begin{propositionrep} 
\label{prop:orbits_permuted}
 Let $\mathcal A = (Q, \Sigma, \delta, q_0, F)$
 be a permutation automaton and $\Gamma \subseteq \Sigma$. Set $\Delta = \Sigma \setminus \Gamma$.
 Then, $\mathcal A$ is a state-partition automaton for $\pi_{\Gamma}$
 if the $\Delta$-orbits of the form $\Orb_{\Delta}(q)$
 are permuted, i.e., for each $x \in \Sigma$ and $q \in Q$, we have
 $
  \delta(\Orb_{\Delta}(q), x) = \Orb_{\Delta}(\delta(q,x)).
 $
\end{propositionrep}
\begin{proof}
 First, set $\Delta = \Sigma \setminus \Gamma$.
 Let $\mathcal R_{\mathcal A}^{\Gamma} = (\mathcal P(Q), \Gamma, \mu, \Orb_{\Delta}(q_0), E)$ be the projection automaton.
 By Lemma~\ref{lem:orbits_for_normal_subgroup},
 from $\Orb_{\Delta}(q_0)$
 only subsets of the form $\Orb_{\Delta}(q)$
 for $q \in Q$, are reachable.
 So, with Lemma~\ref{lem:orbits-partition},
 $\mathcal A$ is a state-partition automaton.
 Note that we have at most $|Q|$ many $\Delta$-orbits, where
 the case that we have precisely $|Q|$ such orbits is only possible
 when, for every letter $x \in \Delta$, the transformation $\delta_x : Q \to Q$
 is the identity transformation.~\qed 
 
\end{proof}

 With Lemma~\ref{lem:orbits-partition}, if the orbits for some $\Delta \subseteq \Sigma$
 are permuted, then, for each $q \in Q$ and $x \in \Sigma$, 
 $
 \delta(\Orb_{\Delta}(q), x) = \Orb_{\Delta}(q)
 $ or
 $
 \delta(\Orb_{\Delta}(q), x) \cap \Orb_{\Delta}(q) = \emptyset.
 $

 \begin{remark}
  The following example shows that $\mathcal A$
  being a state-partition automaton for $\Gamma$ does not imply
  that the sets $\Orb_{\Delta}(q)$ are permuted.
  Let $\mathcal A = (\{1,2,3,4,5,6,7,8\}, \{a,b\}, \delta, 1, \{1\})$
  with the transitions $a = (1,2,3,4)(5,6)(7,8)$ and $b = (1,5)(2,6)(3,7)(4,8)$.
  Then, for $\Gamma = \{b\}$ the automaton is a state-partition automaton, as
  the reachable states in $\mathcal R_{\mathcal A}^{\Gamma}$
  are $\{1,2,3,4\}$ and $\{5,6,7,8\}$, but the $\{a\}$-orbits
  are $\{1,2,3,4\}, \{5,6\}$ and $\{7,8\}$.
 \end{remark}

 Recall that $\mathcal T_{\mathcal A}$ denotes the transformation semigroup of $\mathcal A$.
 A subgroup of $\mathcal T_{\mathcal A}$, if $\mathcal A$ is a finite permutation automaton,
is a subset containing the identity transformation and closed under function composition.
 As we are only concerned with finite automata, this also implies closure under inverse functions.

Next, we show that when the symbols deleted by a projection
generate a normal subgroup, then the automaton is a state-partition automaton
for this projection. 

Normal subgroups are ubiquitous~\cite{cameron_1999,Rotman95} in abstract group theory
as well as in permutation group theory. We give a definition
for subgroups of $\mathcal T_{\mathcal A}$, when $\mathcal A$ is a permutation automaton, using our notation. 
We refer to more specialized literature for other definitions
and more motivation~\cite{cameron_1999,Rotman95}.

\begin{definition}\label{def:normal_subgroup_T_A}
 Let $\mathcal A = (Q,\Sigma, \delta, q_0, F)$ be a permutation automaton.
 Then, a subgroup $N$ of $\mathcal T_{\mathcal A}$
 is called \emph{normal}, if, for each 
 $\delta_u, \delta_v \in \mathcal T_{\mathcal A}$ ($u,v \in \Sigma^*)$,
 \[
  ( \exists \delta_{w} \in N : \delta_u = \delta_{w v} )
  \Leftrightarrow 
  ( \exists \delta_{w'} \in N : \delta_u = \delta_{v w'} ).
 \]

\end{definition} 


If a set of letters generates
a normal subgroup, then the orbits of these letters
are permuted by the other letters. As they are invariant
under the letters themselves that generate these orbits, 
every word over~$\Sigma$ permutes these orbits. This is
the statement of the next lemma.

\begin{lemmarep}
\label{lem:orbits_for_normal_subgroup}
 Let $\mathcal A = (Q, \Sigma, \delta, q_0, F)$ 
 be a permutation automaton and $\Sigma' \subseteq \Sigma$ be
 such that $N = \{ \delta_u : Q \to Q \mid u \in \Sigma'^* \}$
 is a normal subgroup of $\mathcal T_{\mathcal A}$.
 Then, for each $x \in \Sigma$ and $q \in Q$, we have
 $
  \delta(\Orb_{\Sigma'}(q), x) = \Orb_{\Sigma'}(\delta(q,x)).
 $
\end{lemmarep}
\begin{proof}
 Let $q \in Q$ and $x \in \Sigma$.
 If $u \in \Sigma'^*$
 and we consider $\delta_{ux}$, then, by the normality condition,
 there exists $u' \in \Sigma'^*$
 such that $\delta_{ux} = \delta_{xu'}$
 and vice versa.
 Hence, for each $q' \in Q$, applying this, 
 we have
 \begin{equation}\label{eqn:normality_in_T_A}
  \exists u \in \Sigma'^* : q' = \delta(\delta(q, u), x) 
  \Leftrightarrow 
  \exists u' \in \Sigma'^* : q' = \delta(\delta(q, x), u').
 \end{equation}
 So, using Equation~\eqref{eqn:normality_in_T_A},
 \begin{align*}
     q' \in \delta(\Orb_{\Sigma'}(q), x) 
     & \Leftrightarrow \exists u \in \Sigma'^*  : q' = \delta(q, ux) \\ 
     & \Leftrightarrow \exists u' \in \Sigma'^* : q' = \delta(q, xu') \\
     & \Leftrightarrow q' \in \Orb_{\Sigma'}(\delta(q, x)).
 \end{align*}
 Hence,  the transition function $\delta$ maps
 a $\Sigma'$-orbit to another or the same $\Sigma'$-orbit.
 By Lemma~\ref{lem:orbits-partition},
 the $\Sigma'$-orbits partition $Q$.
 So, if $\delta(\Orb_{\Sigma'}(q), x) \cap \Orb_{\Sigma'}(q) \ne \emptyset$,
 we must have $\delta(\Orb_{\Sigma'}(q), x) = \Orb_{\Sigma'}(q)$.
 This finishes the proof.~\qed
\end{proof}


So, combining Proposition~\ref{prop:orbits_permuted}
and Lemma~\ref{lem:orbits_for_normal_subgroup}.

\begin{theorem} 
\label{thm:orbits-normal-subgroup}
 Let $\Gamma \subseteq \Sigma$, $\Delta = \Sigma \setminus \Gamma$
 and $\mathcal A = (Q, \Sigma, \delta, q_0, F)$
 be a permutation automaton.
 Set $N = \{ \delta_{u} : Q \to Q \mid u \in \Delta^* \}$,
 the subgroup in $\mathcal T_{\mathcal A}$ generated by~$\Delta$.
 If $N$ is normal in $\mathcal T_{\mathcal A}$, then $\mathcal A$ is a state-partition automaton
 for $\pi_{\Gamma}$. Hence, in this case, $\pi_{\Gamma}(L(\mathcal A))$
 is recognizable by an automaton with at most $|Q|$
 states. 
\end{theorem}

\section{Commuting Letters} 

Let $\mathcal A = (Q, \Sigma, \delta, q_0, F)$ be a DFA.
We say that two letters $a,b \in \Sigma$ \emph{commute (in~$\mathcal A)$}, if
$\delta(q, ab) = \delta(q, ba)$ for each $q \in Q$. Hence, an automaton $\mathcal A$
is commutative precisely if all letters commute pairwise.

Here, we investigate commuting letters with respect to the projection operation.
Our first lemma states that if we can partition the alphabet of an $n$-state DFA
into two subalphabets
of letters such that each letter in the first set commutes with each letter in the second set,
then for a projection onto one subalphabet,
the projected language is recognizable by an $n$-state automaton.
By this result, the projected language of every $n$-state commutative automaton
is recognizable by an $n$-state automaton.
We construct commutative automata that are not state-partition automata.
Hence, we have new examples of automata whose projected languages are recognizable
by automata with no more states than the original automaton, but which are not state-partition
automata.
Lastly, by investigating the proofs, we can show, with not much more effort, that varieties
of commutative languages are closed under projections.



\begin{lemma} 
\label{lem:centralizer}
 Suppose $\mathcal A = (Q, \Sigma, \delta, q_0, F)$
 is an arbitrary DFA.
 Let $\Gamma \subseteq \Sigma$ be such that, for each $a \in \Sigma\setminus\Gamma$, $b \in \Gamma$ and $q \in Q$,
 we have $\delta(q, ab) = \delta(q, ba)$.
 Then, $\pi_{\Gamma}(L)$ is recognizable by a DFA
 with at most $|Q|$ states.
\end{lemma}
\begin{proof}
 Intuitively, we take the input automaton and leave out all unobservable transitions
 and make a state accepting if, in the input automaton, we can go from this state
 to a final state by a word formed out of the deleted letters.

 Let $\mathcal B = (Q, \Gamma, \delta_{|\Gamma}, q_0, E)$
 be the DFA with $\delta_{|\Gamma}(q, x) = \delta(q,x)$,
 the same start state $q_0$
 and $E = \{ p \in Q \mid \exists q \in F \ \exists u \in (\Sigma \setminus \Gamma)^* : \delta(p, u) = q \}$.
 Then, $L(\mathcal B) = \pi_{\Gamma}(L(\mathcal A))$.

 If $\delta_{|\Gamma}(q_0, u) \in E$,
 then there exists $v \in (\Sigma\setminus\Gamma)^*$
 such that $\delta(q_0, uv) \in F$. So, $uv \in L(\mathcal A)$
 and $u = \pi_{\Gamma}(uv)$.

 Conversely, suppose $u = \pi_{\Gamma}(v)$ for some $v \in L(\mathcal A)$.
 By assumption, as we can successively push all letters in $\Sigma\setminus \Gamma$
 to the end, we have $\delta(q_0, v) = \delta(q_0, \pi_{\Gamma}(v)\pi_{\Sigma\setminus\Gamma}(v))$.
 So, $\delta(q_0, \pi_{\Gamma}(v)\pi_{\Sigma\setminus\Gamma}(v)) \in F$,
 which yields  $\delta_{|\Gamma}(q_0, \pi_{\Gamma}(v)) \in E$,
 hence $u \in L(\mathcal B)$.~\qed
\end{proof}

So, with Lemma~\ref{lem:centralizer}, we get the next result.

\begin{theoremrep}
\label{thm:proj_comm_lang}
 Let $\mathcal A = (Q, \Sigma, \delta, q_0, F)$ be a DFA
 such that $L(\mathcal A)$ is commutative.
 If $\Gamma \subseteq \Sigma$, then $\pi_{\Gamma}(L(\mathcal A))$
 is recognizable by a DFA with at most $|Q|$
 states.
\end{theoremrep}
\begin{proof} 
 First, note that, for each $q \in Q$, $a,b \in \Sigma$ and $u \in \Sigma^*$, we have
 \[
  \delta(q, abu) \in F \Leftrightarrow \delta(q, bau) \in F.
 \]
 Hence, those states are equivalent~\cite{HopUll79} 
 and if $\delta(q, ab) \ne \delta(q, ba)$, we can merge the states
 by identifying them and retaining only the transitions of one state,
 which still gives a recognizing automaton for $L(\mathcal A)$.
 
 Then, apply Lemma~\ref{lem:centralizer}.~\qed
\end{proof}

The definition of normality could be seen
as a generalization of commutativity.
Hence, with Theorem~\ref{thm:orbits-normal-subgroup},
we can deduce the next statement.

\begin{proposition}
\label{cor:perm_com_are_state_partition}
 Let $\mathcal A = (Q, \Sigma, \delta, q_0, F)$
 be a commutative permutation automaton and $\Gamma \subseteq \Sigma$.
 Then, $\mathcal A$ is a state-partition automaton for $\pi_{\Gamma}$.
\end{proposition}

However, there exist commutative automata that are not state-partition automata,
as shown by Example~\ref{ex:comm_proj_non_state_partition}.

\begin{example} 
\label{ex:comm_proj_non_state_partition}
 Let $\mathcal A = (\{q_{\varepsilon}, q_a, q_b\}, \{a,b\}, \delta, q_{\varepsilon}, \{q_{\varepsilon}, q_b\})$
 with  
 \[
  \delta(q_x, y)  = \left\{ \begin{array}{ll}
   q_a & \mbox{if } x = \varepsilon, y = a; \\
   q_b & \mbox{if } x = \varepsilon, y = b; \\ 
   q_b & \mbox{if } x = a, y = b; \\
   q_x & \mbox{otherwise.}
   \end{array}
  \right.
 \]
 Then, $L(\mathcal A) = \{ u \in \{a,b\}^* \mid |u|_a = 0 \mbox{ or } |u|_b > 0 \}$.
 However, we see that $\Orb_{\{b\}}(q_{\varepsilon}) = \{ q_{\varepsilon}, q_b \}$,
 $\Orb_{\{b\}}(q_b) = \{q_b\}$
 and $\Orb_{\{b\}}(q_a) = \{q_a, q_b\}$. Hence,
 $\mathcal A$ is not a state-partition automaton for the projection 
 $\pi_{\{a\}} : \{a,b\}^* \to \{a\}^*$. Also,
 it is not a state-partition automaton for the projection
 onto $\{b\}^*$.
\end{example}

The proofs of Lemma~\ref{lem:centralizer} and Theorem~\ref{thm:proj_comm_lang}
also show that the projected language of a commutative permutation automaton
is recognizable by a permutation automaton, i.e., a group language.
On the other hand, in the general case, Example~\ref{ex:proj_perm_aut}
below gives a permutation automaton whose projected language is not a group language.
Also, most properties defined in terms of automata are preserved by projection
in the commutative case. For example the property of being aperiodic~\cite{Eilenberg1976,Pin86,DBLP:reference/hfl/Pin97}.
We give a more general statement next, showing that many classes from the literature~\cite{Eilenberg1976,Pin86,DBLP:reference/hfl/Pin97} 
are closed under projection when restricted to commutative languages.

%

\begin{theoremrep} 
 Let $\Sigma$ be an alphabet and $\Gamma \subseteq \Sigma$.
 Suppose $\mathcal V$ is a variety of commutative languages. If $L \in \mathcal V(\Sigma^*)$, then $\pi_{\Gamma}(L) \in \mathcal V(\Gamma^*)$.
 In particular, the variety of commutative languages is closed under projection.
\end{theoremrep}
\begin{proof}
 For varieties, the transformation monoid of the minimal complete automaton of a language $L \in \mathcal V(\Sigma^*)$
 is an element of the corresponding variety of monoids~\cite{Eilenberg1976,Pin86}.
 Let $\mathcal A$ be the minimal complete automaton of $L$.
 Then, we see that applying the construction
 of Lemma~\ref{lem:centralizer} also gives a complete automaton $\mathcal B$.
 The transformation monoid $\mathcal T_{\mathcal B}$
 is a submonoid of $\mathcal T_{\mathcal A}$
 by interpreting any $\delta_u \in \mathcal T_{\mathcal B}$, $u \in \Gamma^*$,
 as an element $\delta_u \in \mathcal T_{\mathcal A}$ with $u \in \Sigma^*$.
 By standard constructions~\cite{Eilenberg1976,Pin86}, 
 a transformation monoid recognizes a language if and only if the
 corresponding automaton recognizes the language.
 Hence, by Theorem~\ref{thm:proj_comm_lang},
 the transformation monoid recognizes $\pi_{\Gamma}(L)$.
 As varieties of monoids are closed under submonoids~\cite{Eilenberg1976,Pin86}, 
 we have $\mathcal T_{\mathcal B} \in \mathcal V$ 
 and so $\pi_{\Gamma}(L) \in \mathcal V(\Gamma^*)$.~\qed
\end{proof}

Hence, for example commutative locally-testable, piecewise-testable, star-free or group
languages are preserved under every projection operator, as these classes
form varieties~\cite{Pin86}.


\begin{remark}
 In~\cite{DBLP:journals/tcs/JiraskovaM12} it was stated that 
 languages satisfying the observer property, i.e.,
 that are given by a state-partition automaton for a given projection operator, 
 and the finite languages projected onto unary finite languages
 were the only known languages for which we can recognize the projected language
 with at most the number of states as the original language.
 Note that Theorem~\ref{thm:proj_comm_lang}
 provides genuinely new instances for which this holds true,
 see Example~\ref{ex:comm_proj_non_state_partition}.
\end{remark}

\begin{example}
\label{ex:proj_perm_aut}
Also, for projections, consider the group language given by the permutation
automaton
$\mathcal A = (\{a,b\}, \{0,1,2\}, \delta, 0, \{2\})$
with $a = (0,1)$ and $b = (0,1,2)$.
Then, $\pi_{\{b\}}(L(\mathcal A)) = bb^*$, which is not a group language. 
For example, $b$ is the projection of $ab \in L(\mathcal A)$,
or $bbb$ the projection of $abbab \in L(\mathcal A)$.
\end{example}




\section{Conclusion}

We have continued the investigation of the state complexity of operations
on permutation automata, initiated in~\cite{DBLP:conf/dlt/HospodarM20},
and the investigation of the projection operation~\cite{DBLP:journals/tcs/JiraskovaM12,Wong98}.
We improved the general bound to the tight bound $2^{n - \lceil \frac{m}{2} \rceil} - 1$
in this case.
Note that the general bound $2^{n-1} + 2^{n-m} - 1$ for the projection
is only achieved for automata with precisely $m - 1$ non-loop unobservable transitions~\cite{DBLP:journals/tcs/JiraskovaM12}.
However, if we have more such unobservable transitions, then
it was also shown in~\cite{DBLP:journals/tcs/JiraskovaM12}
that we have the better tight bound $2^{n - 2} + 2^{n-3} + 2^{n - m} - 1$.
For our lower bound stated in Theorem~\ref{thm:lower_bound}
in the case of permutation automata, we have precisely
the same number of non-loop unobservable transitions as states incident
with them. 
Lastly, note that the condition in Proposition~\ref{prop:orbits_permuted}
could be easily checked. Likewise, checking if a subset of letters $\Delta$ generate a normal subgroup
could also be checked efficiently using results from~\cite{DBLP:conf/stoc/BabaiLS87}.

\smallskip \noindent \footnotesize
\textbf{Acknowledgement.} I sincerely thank the anonymous reviewers for careful reading and detailed feedback
that helped me in finding better formulations or fixing typos. Also, 
Section~\ref{sec:normal_subgroups} was restructured after this feedback and the cycle notation
was pointed out to me by one reviewer.

\bibliographystyle{splncs04}
\bibliography{ms} 

\end{document}